\documentclass[12pt,a4paper]{article}
\usepackage[small]{titlesec}

\usepackage[backend=biber,style=numeric,natbib=false,giveninits=true,doi=true,isbn=false,url=false,date=year]{biblatex}
\DeclareNameAlias{default}{last-first}

\DeclareFieldFormat[article]{title}{#1}
\renewbibmacro{in:}{}
\DeclareFieldFormat[article]{pages}{#1}
\DeclareFieldFormat{journal}{#1}
\renewcommand\bf\bfseries

\renewbibmacro*{volume+number+eid}{%
  \printfield{volume}%
  \setunit*{\addnbspace}% NEW (optional); there's also \addnbthinspace
  \printfield{number}%
  \setunit{\addcomma\space}%
  \printfield{eid}}
\DeclareFieldFormat[article]{number}{\mkbibparens{#1}}
\DeclareFieldFormat[article]{volume}{\textbf{#1}}
\DeclareFieldFormat{year}{\mkbibparens{#1}}
\DeclareBibliographyDriver{article}{%
  \printnames{author}:%
  \newunit\newblock
  \printfield{title}%
  \newunit\newblock
  \printfield{journaltitle}
  \newunit
  \iffieldundef{number}{\printfield{volume}}{\printfield{volume}\addspace\printfield{number}}
\addcomma\addspace\printfield{pages}\addspace
  \printfield{year}
  }

\AtEveryBibitem{\clearfield{month}}
\AtEveryBibitem{\clearfield{day}}
\usepackage[latin9]{luainputenc}
\usepackage{color}
\usepackage{amsmath}
\numberwithin{equation}{section}
\usepackage{amsthm}
\usepackage{amssymb}
\usepackage[unicode=true,pdfusetitle,
 bookmarks=true,bookmarksnumbered=false,bookmarksopen=false,
 breaklinks=false,pdfborder={0 0 1},backref=false,colorlinks=false]
 {hyperref}
\usepackage[nameinlink]{cleveref}

\makeatletter
\theoremstyle{plain}
\newtheorem{thm}{\protect\theoremname}
  \theoremstyle{plain}
  \newtheorem{lem}[thm]{\protect\lemmaname}
  \theoremstyle{plain}
  
  \theoremstyle{plain}
  \newtheorem{prop}[thm]{\protect\propositionname}
  \theoremstyle{remark}

  \newtheorem{assumption}{\protect\assumptionname}

  \theoremstyle{remark}
  \newtheorem{rem}{\protect\remarkname}
  \theoremstyle{definition}
  \newtheorem{defn}{\protect\definitionname}
  \theoremstyle{plain}
  \newtheorem{example}{\protect\examplename}
\usepackage{amsthm}
\usepackage{mathrsfs}
\usepackage{bbm}
\usepackage{dsfont} % for fat identity operator
\usepackage{braket}
\usepackage{tikz} % for figures
\usetikzlibrary{arrows} % for arrows in figures

\newcommand{\dif}{\mathrm{d}}
\newcommand{\supp}{\mathrm{supp}}
\newcommand{\im}{\mathrm{im}\,}
\newcommand{\sgn}{\mathrm{sgn}}

\newcommand{\tr}{\operatorname{tr}}
\newcommand{\findex}{\mathrm{ind}\,}

\newcommand{\diag}{\mathrm{diag}\,}
\newcommand{\coker}{\mathrm{coker}\,}
\usepackage{mathtools}
\DeclarePairedDelimiter\norm{\lVert}{\rVert}%
\let\oldnorm\norm
\def\norm{\@ifstar{\oldnorm}{\oldnorm*}}
\makeatother

\title{The Bulk-Edge Correspondence for Disordered Chiral Chains}

\makeatother

  \providecommand{\assumptionname}{Assumption}
  \providecommand{\claimname}{Claim}
  \providecommand{\corollaryname}{Corollary}
  \providecommand{\definitionname}{Definition}
  \providecommand{\lemmaname}{Lemma}
  \providecommand{\propositionname}{Proposition}
  \providecommand{\remarkname}{Remark}
\providecommand{\theoremname}{Theorem}
\providecommand{\examplename}{Example}

\crefname{section}{Section}{Sections}
\crefname{figure}{Figure}{Figures}
\crefname{assumption}{Assumption}{Assumptions}
\crefformat{equation}{(#2#1#3)}

\crefrangelabelformat{equation}{(#3#1#4--#5#2#6)}

\crefmultiformat{equation}{(#2#1#3}{, #2#1#3)}{#2#1#3}{#2#1#3}
\Crefmultiformat{equation}{(#2#1#3}{, #2#1#3)}{#2#1#3}{#2#1#3}

\newtheorem*{lem*}{\protect\lemmaname}

\begin{document}

\author{Gian Michele Graf and Jacob Shapiro\\
 \textit{\normalsize{}Theoretische Physik, ETH Z\"urich, 8093 Zürich,
Switzerland }\\
}
\maketitle
\begin{abstract}
We study one-dimensional insulators obeying a chiral symmetry in the
single-particle picture. The Fermi level is assumed to lie in
a mobility gap. Topological indices are defined for infinite (bulk)
or half-infinite (edge) systems, and it is shown that for a given
Hamiltonian with nearest neighbor hopping the two indices are equal. We also give a new formulation
of the index in terms of the Lyapunov exponents of
the zero energy Schr\"odinger equation, which illustrates the conditions
for a topological phase transition occurring in the mobility gap regime. 
\end{abstract}

\section{Introduction}

Topological materials come in classes differing by symmetry type
and by the dimension of the physical space. The classification table
\cite{Schnyder_Ryu_Furusaki_Ludwig_PhysRevB.78.195125,Schnyder_Ryu_Furusaki_Ludwig_1367-2630-12-6-065010,Kitaev2009,Heinzner2005}
associates an index group to each class or actually a concrete index \cite{PSB_2016, Koma_Katsura_2016arXiv161101928K}, the values of which separate topological phases of materials within
the same class. The fairly general model to be analyzed here obeys
\emph{chiral symmetry} (class AIII of the table) in \emph{dimension
one}, and exhibits moreover \emph{strong disorder}. The symmetry of
the Hamiltonian is matched by that of the state, which is at half-filling. The prototypical
model in the same class, yet lacking disorder, is the Su-Schrieffer-Heeger
model of polyacetylene \cite{SSH_1979}: This is an alternating chain
of sites or, in other words, a bipartite lattice, along which electrons
hop between sub-lattices, either to the right or to the left, but without
experiencing an on-site potential. As a result, the Hamiltonian $H$
and its opposite, $-H$, are unitarily conjugate. In particular the
energy zero is special, being the fixed point under the sign flip,
and it singles out half-filling. If that energy lies in a spectral
gap of $H$, the model exhibits topological properties which depend
on the (constant) ratio of the amplitudes for hopping in the two directions
(from a given sub-lattice). How much of this survives when the hopping changes randomly from
bond to bond? And what if the disorder is actually so strong as to
close the spectral gap about zero? At first sight, disorder seems
to induce localization throughout the spectrum, as it certainly is
the case for on-site randomness \cite{Kunz1980} which corresponds to the class A, and is topologically trivial.
The truth for class AIII however is that localization may fail, but need not, at
the one special energy, i.e. zero. This is enough to rescue the topological
features; in fact Hamiltonians may be loosely viewed as belonging
to a same topological phase as long as they can be deformed while
preserving localization (mobility gap)
at zero energy. Put differently: The closing of the mobility gap about
zero defines the phase boundaries.

More precisely, we will cast the crucial assumption of a mobility
gap in precise technical terms, which relate to known signatures of localization.
We then consider two quantities associated to the bulk and the edge of
the material respectively, and show that they are well-defined and
integer-valued, whence they serve as indices.
We show that they agree (\emph{bulk-edge correspondence}) and finally
that the index can be characterized in terms of the Lyapunov spectrum
of the time-independent Schr\"odinger equation.

The paper is organized as follows. We start in \cref{sec:Chiral-1D-Systems} by describing the mathematical
setting, defining chiral symmetry and its features,
including the bulk and edge invariants. We define the notion of a mobility gap and state the main result about bulk-edge correspondence in that context. We also reformulate the index in terms of Lyapunov exponents. \cref{sec:The-Spectral-Gap-Case} is an aside about the more restrictive case of a spectral gap and the resulting simplifications. For completeness the even more special, translation invariant case is addressed there, too. In \cref{sec:Generalized-States-of-Zero-Energy} we return to the general case by reformulating the edge index, so as to conclude the proof of the main result in \cref{sec:Proof-of-Duality}. \hyperref[sec:The Appendix]{The Appendix} contains a few technical lemmas, as well as a discussion of more general boundary conditions.

In concluding this section we comment on literature on related models formulated in the framework of stochastically translation invariant Hamiltonians \cite{Bellissard_1994JMP....35.5373B}. In \cite{Prodan_Song_2014} a similar model has been discussed in the strong disorder regime and its phases explored numerically; bulk-edge correspondence is shown in \cite{PSB_2016} for the case of the spectral gap. The appropriate bulk index was introduced in \cite{PRODAN20161150} and moreover shown to be well-defined and continuous w.r.t. the Hamiltonian in the case of a mobility gap. Finally we note that in \cite{Mudry_1998_PhysRevLett.81.862} the role of the Lyapunov exponents at zero energy is addressed, including that of a zero exponent in some model.

\section{\label{sec:Chiral-1D-Systems}The model and the results}
\begin{figure}[h]
\centering
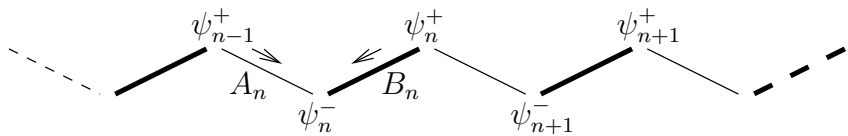
\caption{\label{fig:Chiral zigzag Model}The lattice underlying the model is an alternating chain. The hopping amplitudes $A_n, B_n\in GL_N(\mathbb{C})$ are in direction of the arrows. In the opposite direction the adjoint matrices apply.}
\end{figure}

In this section we shall specify the setting of chiral one-dimensional systems, define the relevant indices, and formulate the main result on bulk-edge duality.

\subsection{One-dimensional chiral systems}
The lattice underlying the model is an alternating chain, where particles perform nearest-neighbor hopping (see \cref{fig:Chiral zigzag Model}). The single-particle Hilbert space of a tight-binding model is 
\[
\mathcal{H}=\mathcal{K}\otimes\mathbb{C}^{2}\ni\begin{pmatrix}\psi_{n}^{+}\\
\psi_{n}^{-}
\end{pmatrix}_{n\in\mathbb{Z}}\,,
\]
with $\mathcal{K}:=\ell^{2}(\mathbb{Z},\,\mathbb{C}^{N})$, where $\mathbb{C}^{N}$ stands for the internal degrees of freedom of each site and $\mathbb{C}^{2}$ for their grouping into dimers. The Hamiltonian is 
\begin{align}
H & = \begin{pmatrix}0 & S^\ast\\
S & 0
\end{pmatrix}
\label{eq:chiral Hamiltonian}
\end{align}
with $S$ acting on $\mathcal{K}$ as
\begin{align}
(S\psi^{+})_{n} & :=  A_{n}\psi_{n-1}^{+}+B_{n}\psi_{n}^{+}\,;\label{eq:action of S in our model}
\end{align}
hence 
\begin{align}
(S^{\ast}\psi^{-})_{n} & = A_{n+1}^{\ast}\psi_{n+1}^{-}+B_{n}^{\ast}\psi_{n}^{-}\,.\label{eq:action of S* in our model}
\end{align}
We assume $A_{n},\,B_{n}\in GL_N(\mathbb{C})$, whence solutions to $S\psi^+=0$ are determined by $\psi^+_n$ for any $n$. Otherwise,
i.e. if some matrices were singular, the corresponding bonds would
be effectively broken; put differently, the model would have edges
within.

The chiral symmetry 
\begin{align*}
\Pi & :=  \begin{pmatrix}\mathds{1} & 0\\
0 & -\mathds{1}
\end{pmatrix}
\end{align*}
is a symmetry of the Hamiltonian, in the sense that 
$$
\{H,\,\Pi\}  \equiv H \Pi + \Pi H = 0\,.
$$
It implies
\begin{align}
f(H)\Pi = \Pi f(-H)
\label{eq:If H anti-commutes with chirality so does its odd functional calculus}
\end{align}
for any Borel bounded function $f=f(\lambda)$.

The many-particle state is the Fermi sea at half-filling, meaning that the Fermi level is at $\lambda=0$. Its single-particle density matrix thus is the Fermi projection $P:=\chi_{(-\infty,0)}(H)$, where $\chi_I$ is the characteristic function of the set $I\subseteq\mathbb{R}$.

We assume localization for $H$ at the Fermi level and formulate that condition deterministically. In a companion paper \cite{Graf_Shapiro_2017} it will be shown that it holds either with probability zero or one, depending on the details of the model. This means that no further recourse to probabilistic arguments will be made in proofs, and that the indices are properties of the individual system, and not just of the statistical ensemble.
\begin{assumption}
\label{assu:exp_decay_of_Fermi_projection}For some $\mu,\,\nu>0$ we have
\begin{align*}
\sum_{n,\,n'\in\mathbb{Z}}\|P(n,n')\|(1+\left|n\right|)^{-\nu}e^{\mu\left|n-n'\right|} & \leq C < +  \infty\,,
\end{align*}
where $(\delta_n)_{n\in\mathbb{Z}}$ is the canonical (position) basis of $\ell^2(\mathbb{Z})$ and the map $P(n,n')=\langle\delta_n,\,P\delta_{n'}\rangle:\mathbb{C}^{2N}\to\mathbb{C}^{2N}$ acts between the internal spaces of dimers of $n'$ and $n$. Here, $\|\cdot\|$ is the trace norm of such maps. Moreover, the same bound applies to the Fermi projections of the edge Hamiltonians introduced below.
\end{assumption}

\begin{assumption}
\label{assu:zero_is_not_an_eigenvale_of_the_Hamiltonian}$\lambda=0$ is not an eigenvalue of $H$.
\end{assumption}

\begin{rem}
These two assumptions are trivially fulfilled in the spectral gap case. In this paper we are rather interested in the mobility gap regime, which is the typical one at large disorder.

\end{rem}

In physical terms \cref{assu:zero_is_not_an_eigenvale_of_the_Hamiltonian} states that every state is either a particle or a hole state, thus prompting the notation 
\begin{align*}
P_{-}:=P\,,\qquad P_{+}:=\chi_{(0,\infty)}(H) 
\end{align*}
and the rephrasing
\begin{align}
P_{-}+P_{+}=\mathds{1}\,,\qquad P_{-}\Pi=\Pi P_{+}\label{eq:rephrasing of the assumption of chiral symmetry}
\end{align}
of the assumption and of the chiral symmetry.

We will define shortly a \emph{bulk index} $\mathcal{N}$ associated to $H$,
as well as an \emph{edge index} $\mathcal{N}_{a}$ associated to its
truncation to the half-lattice to the left of an arbitrary point $a\in\mathbb{Z}$.
\begin{thm}
\label{thm:BEC}(Duality) Under  \cref{assu:zero_is_not_an_eigenvale_of_the_Hamiltonian,assu:exp_decay_of_Fermi_projection} we have
\begin{align*}
\mathcal{N} & =  \mathcal{N}_{a}\,.
%\label{eq:bulk-edge-duality}
\end{align*}
\end{thm}

We anticipate that $\mathcal{N}_{a}$ will be manifestly an integer.
Hence so is $\mathcal{N}$, and $\mathcal{N}_{a}$ is independent
of $a$. In the proof though, we will first establish the independence and then obtain the result by passing to the limit $a\to+\infty$. The two steps will be carried out in \cref{sec:Generalized-States-of-Zero-Energy,sec:Proof-of-Duality}.

\paragraph{Bulk Index.}
Let $\Sigma := \sgn\, H$. Let $\Lambda:\mathbb{Z}\to\mathbb{R}$ be a \emph{switch function}, i.e.
$\Lambda\left(n\right)=1$ (resp. $=0$) for $n$ (resp. $-n$) large
and positive. It defines a (multiplication) operator on $\ell^{2}(\mathbb{Z})$,
which carries naturally to its descendant spaces $\mathcal{K}$ and
$\mathcal{H}$.
\begin{defn}The bulk index is
\begin{align}
\mathcal{N} & :=  \frac{1}{2}\tr(\Pi \Sigma[\Lambda,\,\Sigma])\label{eq:Bulk topological invariant}\,.
\end{align}
\end{defn}

The index is well-defined. In fact, $\Sigma=P_{+}-P_{-}$ with $P_{\pm}$ as above and we have
\begin{lem}
\label{lem:=commutator of L and P is trace class}$[\Lambda,\,P_\pm]$
are trace class and the index can be expressed as \begin{align}\label{eq:reformulation of bulk invariant}\mathcal{N} &= -\tr\Pi P_{+} [\Lambda,P_-] - \tr\Pi P_- [\Lambda,P_{+}]\,.\end{align}
\end{lem}

\paragraph{Edge Index.}

The model is truncated to
$\mathbb{Z}_{a}:=(-\infty,a]\subset\mathbb{Z}$ with Hilbert space
$\mathcal{H}_{a}:=\ell^{2}(\mathbb{Z}_{a},\,\mathbb{C}^{2N})$. Of
course the choice of $a$ ought not be of physical relevance. Here we
keep this choice free and explicit in the notation since we shall
eventually take the limit $a\to+\infty$ which helps associating edge
objects with bulk ones.

The truncation procedure can be recast algebraically as follows. Let $\iota_{a}:\ell^{2}(\mathbb{Z}_{a})\hookrightarrow\ell^{2}(\mathbb{Z})$ be the natural injection, whence $\iota_{a}^{\ast}:\ell^{2}(\mathbb{Z})\twoheadrightarrow\ell^{2}(\mathbb{Z}_{a})$ is the restriction operator. Thus $\iota_{a}$ is an isometry, but not a unitary: In fact 
\begin{align}
\iota_{a}^{\ast}\iota_{a}  =  \mathds{1}_{\mathcal{H}_{a}}\label{eq:partial isometry condition}&\,,\qquad\iota_{a}\iota_{a}^{\ast}=\chi_a\,,
\end{align}
where $\chi_a$ is the projection $\chi_a:\ell^{2}(\mathbb{Z})\to\ell^{2}(\mathbb{Z})$ associated to the subspace $\ell^2(\mathbb{Z}_{a})\subseteq\ell^2(\mathbb{Z})$. Thus 
\begin{align*}
\chi_a\iota_a = \iota_a\,&,\qquad\iota_a^\ast=\iota_a^\ast\chi_a
%\label{eq:relations with iota and chi}
\end{align*}
and
\begin{align}
\chi_a A \chi_a:\im \chi_a \to \im \chi_a = \iota_a^\ast A \iota_a 
\label{eq:relations with iota and chi_a}
\end{align}
for any operator $A$ acting on $\ell^{2}(\mathbb{Z})$ or on some of its descendant spaces. In particular, letting
\begin{align}
S_a := \iota_a^\ast S \iota_a\,,
\label{eq:relation between S_a and S}
\end{align}
we then have 
\begin{align}
H_{a}  &:= \iota_a^\ast H\iota_a =  \begin{pmatrix}0 & S_{a}^\ast\\
S_{a} & 0
\end{pmatrix}\,.\label{eq:chiral edge Hamiltonian}
\end{align}
More general boundary conditions will be discussed in \cref{subsec:General B.C.}.

\begin{rem}
As $a\to+\infty$ an ever larger portion $\mathbb{Z}_a$ of $\mathbb{Z}$ is retained, resulting in the limit $H_a\to H$ in the strong resolvent sense, as will be seen and used.
\end{rem}

From \cref{eq:chiral edge Hamiltonian} we still have 
\begin{align}
\{H_{a},\,\Pi\} & =  0\label{eq:chiral symmetry constraint for edge Hamiltonian}
\end{align}
and its consequence \cref{eq:If H anti-commutes with chirality so does its odd functional calculus}.

Let $P_{0,a}:=\chi_{\{0\}}(H_{a})$ be the spectral projection
for $\lambda=0$, i.e. its eigenprojection if it is an eigenvalue. We note that \cref{assu:zero_is_not_an_eigenvale_of_the_Hamiltonian} generically fails for the edge system.
\begin{defn}
The edge index is
\begin{align}
\mathcal{N}_{a} & :=  \tr(\Pi P_{0,a})\,.\label{eq:edge invariant}
\end{align}
\end{defn}

\begin{rem}
$\Pi$ maps $\im P_{0,a}=\ker H_{a}$ into itself. Indeed, $H_{a}\psi=0$
implies $H_{a}\Pi\psi=0$ by \cref{eq:chiral symmetry constraint for edge Hamiltonian}. In particular $\mathcal{N}_{a}\in\mathbb{Z}$
as anticipated, and the index may be written as 
\begin{align}
\mathcal{N}_{a} & =  \dim\ker S_{a}-\dim\ker S_{a}^{\ast}\label{eq:Fredholm index of S_a}\,,
\end{align}
which is finite by \cref{eq:action of S in our model,eq:action of S* in our model}. Despite appearances, this is not a Fredholm index in general, simply because $S_a$ is not Fredholm in the mobility gap regime of \cref{assu:exp_decay_of_Fermi_projection}. Indeed, $\im S_a$ is not closed then.
\end{rem}

\begin{example}
\cref{fig:Chiral zigzag Model} should be viewed as just one example of a lattice leading to a chiral Hamiltonian \cref{eq:chiral Hamiltonian}. Other lattices may do so too. An example is shown in \cref{fig:Example for another chiral structure}. 
\begin{figure}[b]

\begin{center}
\begin{tikzpicture}[scale=1, transform shape]
% Helper grid (only when editing)
%\draw[step=1cm,gray,very thin] (0,2) grid (11,6); 

% Define the nodes (the sites of the lattice) 
\node at (1,3) (phim2m) {$\varphi_{m-2}^-$};
\node at (1,5) (phim2p) {$\varphi_{m-2}^+$};

\node at (4,3) (phim1m) {$\varphi_{m-1}^-$};
\node at (4,5) (phim1p) {$\varphi_{m-1}^+$};

\node at (7,3) (phim) {$\varphi_{m}^-$};
\node at (7,5) (phip) {$\varphi_{m}^+$};

\node at (10,3) (phip1m) {$\varphi_{m+1}^-$};
\node at (10,5) (phip1p) {$\varphi_{m+1}^+$};

% Vertical lines
\draw[ thick]  (phim2m) -- (phim2p);
\path[->, thick] (phim1p) edge node [right] {\tiny$ T_{m-1}$} (phim1m);
\path[->, thick] (phip) edge node [right] {\tiny$T_m$} (phim);
\draw[thick]  (phip1m) -- (phip1p);
% Diagonal lines
\draw[thick] (phim2p) -- (phim1m);
\path[->,thick] (phim1p) edge node [above=0.1cm] {\tiny$T_{m-1}^-$} (phim2m);
\path[->,thick] (phim1p) edge node [above=0.3cm] {\tiny$T_{m-1^+}$} (phim);
\path[->,thick] (phip) edge node [right=0.1cm] {\tiny$T_m^-$}  (phim1m);
\path[->,thick] (phip) edge node [left=0.1cm] {\tiny$T_m^+$}(phip1m);
\draw[thick] (phim) -- (phip1p);
% dashed lines to infinity
\draw[dashed,thick] (phim2p) -- (-0.5,4);
\draw[dashed,thick] (phim2m) -- (-0.5,4);
\draw[dashed,thick] (phip1p) -- (11.5,4);
\draw[dashed,thick] (phip1m) -- (11.5,4);
% dashed dimers
\draw[thick, dashed, rounded corners=10pt]
  (3.4,2.6) rectangle (7.6,3.4);
  \draw[thick, dashed, rounded corners=10pt]
  (6.4,4.6) rectangle (10.6,5.4);
\end{tikzpicture}
\end{center}

\caption{\label{fig:Example for another chiral structure}A lattice with hopping amplitudes $T_m,T_m^{\pm}\in GL_M(\mathbb{C})$ in direction of the arrows. The blobs indicate a regrouping of the sites relating the lattice to that of \cref{fig:Chiral zigzag Model}.}
\end{figure}
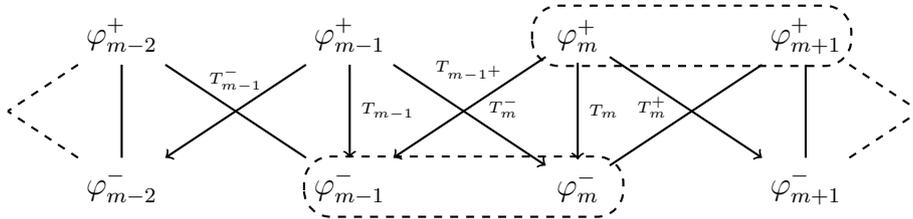

The model is of the form \cref{eq:chiral Hamiltonian,eq:action of S in our model} with $N=2M$ upon grouping amplitudes as bispinors: 
$$ \psi_n^-=\begin{pmatrix} \varphi_{2n-1}^-\\
 \varphi_{2n}^-
\end{pmatrix},\quad\psi_n^+=\begin{pmatrix} \varphi_{2n}^+\\
\varphi_{2n+1}^+
\end{pmatrix}\,.$$
Comparing \cref{fig:Chiral zigzag Model,fig:Example for another chiral structure} yields $$ A_n = \begin{pmatrix}T_{2n-2}^+ & T_{2n-1}\\ 0 & T_{2n-1}^+\end{pmatrix}\,,\quad B_n = \begin{pmatrix}T_{2n}^- & 0 \\ T_{2n}&T_{2n+1}^-\end{pmatrix}\,.$$
In particular $A_n, B_n \in GL_N(\mathbb{C})$ iff $T_m^\pm\in GL_M(\mathbb{C})$.
\end{example}

\subsection{The zero-energy Lyapunov spectrum\label{sec:Formulation-Via-Lyapunov-Exponents}}
We conclude this section with an alternate formulation of the index. To this end we consider the equation $S\psi^+=0$ as a finite difference equation for sequences $\psi^+:\mathbb{Z}\to\mathbb{C}^N$, foregoing normalizability. By \cref{eq:action of S in our model} and $A_n\in GL_N(\mathbb{C})$ the equation is solved recursively,
\begin{align}
S\psi^+=0\qquad\Longleftrightarrow\qquad\psi_{n-1}^+=T_n\psi_n^+\,,\qquad(n\in\mathbb{Z})
\label{eq:finite difference equation formally extended to non normalizable vectors}
\end{align}
with $T_n:=-A_n^{-1}B_n$. The associated transfer matrix is
\begin{align*}
T(n):=T_{n-1}\cdots T_0\,,\qquad(n<0)\,.
\end{align*}
The Lyapunov exponent of a vector $v\in\mathbb{C}^N$ is then given as
\begin{align}
\chi(v):=\limsup_{n\to-\infty}\frac{1}{|n|}\log\|T(n)v\|
\label{eq:the Lyapunov exponent}
\end{align}
with $\chi(v)\in\bar{\mathbb{R}}\equiv\mathbb{R}\cup\{\pm\infty\}$ and $\chi(0)=-\infty$. The set 
\begin{align}
V_\chi &:=\{v\in\mathbb{C}^N | \chi(v)\leq\chi\}
\label{eq:filtration of C^N by the Lyapunov spectrum}
\end{align}
is a linear subspace which is non-decreasing in $\chi\in\bar{\mathbb{R}}$. Let $\chi_N\leq\cdots\leq\chi_1$ be the values of $\chi$ at which $\chi\mapsto\dim V_\chi$ jumps, listed repeatedly according to the jump in dimension.
\begin{assumption}
Let $0$ not be in the Lyapunov spectrum, i.e. $\chi_i\neq0$, ($i=1,\dots,N$).
\label{assu:zero is not in the Lyapunov spectrum}
\end{assumption}
\begin{thm} \label{thm:reformulation of edge index in terms of the Lyapunov spectrum}Under \cref{assu:zero is not in the Lyapunov spectrum} the edge index equals the number of negative Lyapunov exponents: 
\begin{align}
\mathcal{N}_a=\#\{i|\chi_i<0\}
\label{eq:reformulation of edge index using Lyapunov spectrum}
\end{align}
for any $a\in\mathbb{Z}$.
\end{thm}
The proof will be given in \cref{sec:Generalized-States-of-Zero-Energy}. In \cite{Graf_Shapiro_2017} we will give conditions such that in \cref{eq:the Lyapunov exponent} $\limsup$ can be replaced almost surely by $\lim$ for all $v\in\mathbb{C}^N$, actually with the limits being finite for $v\neq0$ and with simple Lyapunov spectrum. Moreover, $V_\chi$ is the spectral subspace of the self-adjoint matrix $\Lambda := \lim_{n\to\infty}(T(n)^\ast T(n))^{1/2n}$ and of eigenvalues $\leq e^\chi$. Finally, \cref{assu:exp_decay_of_Fermi_projection,assu:zero is not in the Lyapunov spectrum} will be shown to be equivalent.
\begin{rem}
As a complement to \cref{eq:reformulation of edge index using Lyapunov spectrum}, the edge index $\mathcal{N}_a$ may also be expressed in terms of the equation $S^*\psi^-=0$, in which case it is given by the number of positive Lyapunov exponents. In fact, introducing $\varphi_n^-=B_n^\ast\psi_n^-$ that equation is $\varphi_{n-1}^-={\tilde T}_n\varphi_n^-$, where $\tilde{T}_n=T^\circ_n$ and $M^\circ=(M^\ast)^{-1}$ (using $B_n\in GL_N(\mathbb{C})$, too). Its spaces are 
\begin{align}
\tilde{V}_\chi = V_{-\chi}^\perp\,,
\label{eq:Oseledet spaces of the circ sequence}
\end{align}
provided $\chi$ is not in the Lyapunov spectrum. In particular, $\tilde{\chi}_i=-\chi_{N+1-i}$. Eq. \cref{eq:Oseledet spaces of the circ sequence} follows from $\tilde{T}(n)=T(n)^\circ$ and $\tilde{\Lambda}=\Lambda^{-1}$.
\end{rem}
\begin{rem}
The usual scenario of a phase transition is that of a spectral gap closing on the Fermi level. \cref{thm:reformulation of edge index in terms of the Lyapunov spectrum} gives a different scenario, whereby the Fermi level may not lie in a gap throughout the transition. More precisely, the Lyapunov spectrum associated to $(H-\lambda)\psi=0$ consists of $2N$ exponents $\{\gamma_i\}_{i=1}^{2N}$ and is even under sign flip $\gamma\mapsto-\gamma$ (counting multiplicity). For $\lambda\neq0$ the spectrum is moreover simple, implying that $0$ is not an exponent and thus localization. For $\lambda=0$ however the exponents are those of $S\psi^+=0$ and their flips, $\{\chi_i,-\chi_i\}_{i=1}^N$. In particular $0$ may, but need not be an exponent. If it isn't, \cref{thm:reformulation of edge index in terms of the Lyapunov spectrum} applies, but if it becomes one, the localization length diverges at $\lambda=0$, signaling the topological phase transition.
\end{rem}

\section{\label{sec:The-Spectral-Gap-Case}The spectral gap case}
In the case of a spectral gap the analysis simplifies, as was noted in \cite{PSB_2016} and discussed in terms of K-theory. We present here an equivalent simplification as a contrast to the general case, to be proven later. Until the end of this section we forgo definition \cref{eq:action of S in our model} and allow $S$ to be any operator $S:\mathcal{K}\to\mathcal{K}$ for which $[\Lambda, S]$ is trace class; this being a generalization since the commutator is of finite rank in the former case. The spectral gap condition means that \cref{assu:exp_decay_of_Fermi_projection,assu:zero_is_not_an_eigenvale_of_the_Hamiltonian} are now replaced by the stronger condition  \begin{align} 0&\notin \sigma(H)\,.\label{eq:Spectral Gap Condition}\end{align}
Thus $H$ is Fredholm, and so is $S$ in view of \begin{align*} \ker H &= \ker S \oplus \ker S^\ast\,,\qquad \im H = \im S^\ast \oplus \im S\,.
\end{align*}
We first discuss the bulk index:
\begin{lem} The index \cref{eq:Bulk topological invariant} is well-defined and equals \begin{align}\mathcal{N} = \tr U^\ast [ \Lambda,U]\,,\label{eq:Reformulation of bulk index for spectral gap BEC proof}\end{align} where $U$ is the (unique) unitary in the polar decomposition of $S$: $S=U|S|$ with $|S|\equiv(S^\ast S)^{1/2}$.
\end{lem}
\begin{proof}By \cref{eq:Spectral Gap Condition}, $|H|$ is invertible and we have $\Sigma = H|H|^{-1}$. That operator is computed as \begin{align*}\Sigma = \begin{pmatrix}0 & U^\ast\\
U & 0
\end{pmatrix}\,,\end{align*} as seen from $S^\ast = |S|U^\ast$, $H^2 = \diag(S^\ast S, SS^\ast) = \diag(|S|^2, U|S|^2U^\ast)$, $|H|=\diag(|S|,U|S|U^\ast)$. We conclude that \begin{align*}[\Lambda,\Sigma] = \begin{pmatrix}0 & [\Lambda,U^\ast]\\
[\Lambda,U] & 0
\end{pmatrix}\end{align*} and $\Pi \Sigma [\Lambda, \Sigma] = U^\ast[\Lambda,U]\oplus[\Lambda,U]U^\ast$. The claim is now immediate, provided $[\Lambda,U]$ is trace class. This holds true by the following lemma, because $[\Lambda, S]$ already is trace class.
\end{proof}
\begin{lem} Let $A:\mathcal{K}\to\mathcal{K}$ be Fredholm. If $[A,\Lambda]$ is trace class, then so is $[A|A|^{-1},\Lambda]$.
\end{lem}
\begin{proof}
The commutator property is inherited under taking adjoints and products; and if $A\geq\varepsilon>0$ also under taking inverses, $[A^{-1},\Lambda]=-A^{-1}[A,\Lambda]A^{-1}$. In the latter case the property is also passed down to $A^{-1/2}$ because of \begin{align*}A^{-1/2} = C \int_0^{\infty} \lambda^{-1/2} (A+\lambda)^{-1}\dif{\lambda}\end{align*}($C^{-1} = \int_0^{\infty} \lambda^{-1/2} (1+\lambda)^{-1}\dif{\lambda}$). In particular, the property applies to $A^\ast A = |A|^2$ and to $|A|^{-1}$, where we used that $A$ is Fredholm through $A^\ast A \geq \varepsilon > 0$; finally it applies to $A|A|^{-1}$.
\end{proof} 

We notice that the index \cref{eq:Reformulation of bulk index for spectral gap BEC proof} is independent of the choice of the switch function, this being tantamount to the vanishing of the expression when $\Lambda$ is replaced by a function of compact support. Then, in fact, $\Lambda$ would already be trace class and the claim seen by expanding the commutator. In particular we may pick $\Lambda=\mathds{1}-\chi_a$, $\chi_a$ being the projection seen in \cref{eq:partial isometry condition}. 
We then conclude by (\cite{Avron1994220}, Theorems 6.1, 5.2) that the bulk index is that of a pair of projections: \begin{align*}\mathcal{N} &= -\tr(U^\ast\chi_a U-\chi_a)\\ &=\findex(\chi_a,U^\ast \chi_a U) = \findex(\chi_a U \chi_a:\im\chi_a\to\im\chi_a) \\ &= \findex(\iota_a^\ast U \iota_a)\,.\end{align*} 

We now turn to the edge index and first state a definition: Let \begin{align*}\sigma_{ess}(A)=\{\lambda\in\mathbb{C}|A-\lambda\mathds{1}\text{ is not Fredholm}\}\end{align*} be the essential spectrum of a (not necessarily self-adjoint) closed operator $A$. It enjoys stability under compact perturbations $K$, i.e. $\sigma_{ess}(A)=\sigma_{ess}(A+K)$, \cite{Booss_Topology_and_Analysis}.
\begin{lem}\label{lem:essential spectrum of half line operator is in essential spectrum of whole operator}
Suppose $A:\mathcal{K}\to\mathcal{K}$ is such that $[A,\Lambda]$ is compact, where $\Lambda$ is some (and hence any) switch function. Then \begin{align}\sigma_{ess}(\iota_a^\ast A\iota_a)\subset\sigma_{ess}(A)\,.\label{eq:essential spectrum of half line operator is in essential spectrum of whole operator}\end{align} In particular, if $A$ is Fredholm, then so is $\iota_a^\ast A\iota_a$.
\end{lem}
\begin{proof}
We have \begin{align}A=\Lambda A \Lambda + (\mathds{1}-\Lambda) A (\mathds{1}-\Lambda) + K\label{eq:Decomposition of A in spectral gap BEC}\end{align} with $K = \Lambda A (\mathds{1}-\Lambda) + (\mathds{1}-\Lambda)A \Lambda = [\Lambda,A](\mathds{1}-\Lambda) + (\mathds{1}-\Lambda)[A,\Lambda]$ compact.

The injection $\iota_a$ is matched by another one, $\tilde{\iota}_a$, corresponding to the complementary half-line $\mathbb{Z}\setminus \mathbb{Z}_a$. For $\Lambda=\chi_a$ the first two terms on the RHS of \cref{eq:Decomposition of A in spectral gap BEC} are \begin{align*}\chi_a A \chi_a + (\mathds{1}-\chi_a) A (\mathds{1}-\chi_a)\cong\iota_a^\ast A \iota_a \oplus \tilde{\iota}_a^\ast A \tilde{\iota}_a\end{align*} because of \cref{eq:relations with iota and chi_a} and the unitarity of $\iota_a\oplus\tilde{\iota}_a$. Eq. \cref{eq:essential spectrum of half line operator is in essential spectrum of whole operator} follows.
\end{proof}

For $A=S$ the lemma yields that $S_a$ is Fredholm by \cref{eq:relation between S_a and S}. Thus $\im S_a$ is closed, $\ker S_a^\ast = \coker S_a$, and the edge index \cref{eq:Fredholm index of S_a} is a Fredholm index, \begin{align}\mathcal{N}_a = \findex S_a\,.\label{eq:Edge index is a Fredholm index}\end{align} The proof of the bulk-edge duality, $\mathcal{N}=\mathcal{N}_a$, is concluded by 

\begin{lem}
\begin{align}
\findex \iota_a^\ast S \iota_a = \findex  \iota_a^\ast U \iota_a\,.
\label{eq:BEC for spectral gap}\end{align}
\end{lem}
\begin{proof}
We consider the interpolating family \begin{align*}S_t = t U + (1-t)S = U(t\mathds{1}+(1-t)|S|)\,,\end{align*} ($0\leq t\leq1$), with $S_0 = S$, $S_1 = U$. The hypothesis of \cref{lem:essential spectrum of half line operator is in essential spectrum of whole operator} holds true for $t=0,1$, as remarked before, and thus true for $0\leq t \leq 1$. Moreover, by that lemma, $\iota_a^\ast S_t\iota_a$ is Fredholm if $S_t$ is. That however is immediate from \begin{align*}S_t^\ast S_t = (t\mathds{1}+(1-t)|S|)^2 \geq \delta\,,\end{align*} for some $\delta>0$. Thus \cref{eq:BEC for spectral gap} holds true by the continuity of the index.
\end{proof}

\paragraph{The translation invariant case.}
We here assume that $S:\mathcal{K}\to\mathcal{K}$ commutes with the shift operator, whence $S$ is of Toeplitz form in the position basis $(\delta_n)_{n\in\mathbb{Z}}$ of $\ell^2(\mathbb{Z})$, $$\langle\delta_n, S\delta_{n'}\rangle=S_{n-n'}\,,$$ where the maps $S_m:\mathbb{C}^N\to\mathbb{C}^N$ may themselves be viewed as matrices. For simplicity we assume that they rapidly decay in $m\in\mathbb{Z}$.

\begin{prop}\label{prop:BEC in translation invariant case}The sum $$S(z):=\sum_{m\in\mathbb{Z}}S_m z^{-m}$$ is absolutely convergent for $|z|=1$, i.e. for $z$ on the unit circle $\mathcal{C}$. Then the spectral gap condition \cref{eq:Spectral Gap Condition} holds iff $\det S(z)$ vanishes nowhere on $\mathcal{C}$. In that case the index \cref{eq:Reformulation of bulk index for spectral gap BEC proof} is (the negative of) the winding number of $\mathcal{C}\ni z\mapsto \det S(z)\in\mathbb{C}$ (Zak number \cite{Zak_PhysRevLett.62.2747}).\end{prop}

\begin{example}\label{exam:translation invariant example}
In the translation invariant case, \cref{eq:action of S in our model} reduces to $S(z)=Az^{-1}+B$. The spectral gap condition requires that \begin{align}T:=-A^{-1}B\label{eq:definition of T}\end{align} has no eigenvalue of unit modulus, or equivalently that $1$ is not among its singular values. Since $z\mapsto w=z^{-1}$ reverses the orientation of $\mathcal{C}$, the index equals the winding number of $w\mapsto\det(Aw+B)$ and, by the argument principle, the number of zeroes in $\mathcal{C}$, i.e. the algebraic number of eigenvalues $w$ of $T$ with $|w|<1$; equivalently it is the number of zeroes of $\det S(z)$ with $|z|>1$. 

\end{example}
\begin{proof}[Proof of \cref{prop:BEC in translation invariant case}]
In line with the general assumptions of this section, we first verify that $[\Lambda, S]$ is trace class. This follows from $$\langle\delta_n,[\Lambda, S]\delta_{n'}\rangle=(\Lambda(n)-\Lambda(n'))S_{n-n'}\,,$$ and from \cref{eq:bound on trace class norm} by the reasoning used in the proof of \cref{lem:=commutator of L and P is trace class}. Second, we discuss the gap condition \cref{eq:Spectral Gap Condition}: By Bloch decomposition, 
\begin{align}
S=\int_{\mathcal{C}}^{\oplus}S(z)\frac{\dif{s}}{2\pi}
\label{eq:Bloch decomposition of S}
\end{align}
with $\dif{s}=|\dif{z}|=-iz^{-1}\dif{z}$ and w.r.t. $\mathcal{K}=\ell^2(\mathbb{Z})\otimes\mathbb{C}^N$, $\ell^2(\mathbb{Z})\cong\int_{\mathcal{C}}^\oplus \mathbb{C}\, \dif{s}/2\pi$. The isomorphism is given by $$ \psi(z)=\sum_{n\in\mathbb{Z}} z^{-n}\psi_n $$ with Parseval identity 
\begin{align}
\sum_{n\in\mathbb{Z}} \varphi_n^\ast\psi_n = \int_{\mathcal{C}} \varphi(z)^\ast \psi(z) \frac{\dif{s}}{2\pi}\,.
\label{eq:Parseval identity}
\end{align} 
Moreover $(S\psi)(z)=S(z)\psi(z)$ is readily verified, proving \cref{eq:Bloch decomposition of S}. Since $S(z)$ is smooth we have $\sigma(H)=\bigcup_{z\in\mathcal{C}}\sigma(H(z))$. The claim on the spectral gap property now follows, and we assume its validity in the sequel. 

Next we compute the index \cref{eq:Reformulation of bulk index for spectral gap BEC proof}. The fibers $U(z)=S(z)/|S(z)|$ are smooth as well, whence the sum $$U(z)=:\sum_{m\in\mathbb{Z}}U_m z^{-m}$$ has rapidly decaying coefficients $U_m$. Thus 
\begin{align*}
\mathcal{N} &= \tr U^\ast[\Lambda, U]= \sum_{n\in\mathbb{Z}} \tr\langle U\delta_n, [\Lambda, U]\delta_n\rangle\\
&= \sum_{n,m}\tr |U_{m-n}|^2(\Lambda(m)-\Lambda(n))= \sum_{n,k}\tr|U_k|^2(\Lambda(n+k)-\Lambda(n))\\
&= \sum_{k} k\tr|U_k|^2\,,
\end{align*}
where we used $\sum_n \Lambda(n+1)-\Lambda(n)=1$. Using $-z\partial_z U(z) = \sum_{m} m U_m z^{-m}$ and \cref{eq:Parseval identity} we obtain
\begin{align*}
\mathcal{N} &= \frac{i}{2\pi} \int_{\mathcal{C}} \dif{z}\, \tr U(z)^\ast \partial_z U(z) \\
&=  \frac{i}{2\pi}  \int_{\mathcal{C}} \dif{z}\, \frac{\frac{\dif{\,}}{\dif{z}}\det U(z)}{\det U(z)} =  \frac{i}{2\pi}  \int_{\mathcal{C}} \frac{\dif{\,\det S(z)}}{\det S(z)}\,,
\end{align*}
because $\det |S(z)| > 0$ has no winding.
\end{proof}

\section{\label{sec:Generalized-States-of-Zero-Energy}Generalized states
of zero energy}
Zero energy edge states will be extended to bulk states which are however not $\ell^2$ on the other side of the edge. For that purpose, let
us consider the (bulk) equation $S\psi^{+}=0$ as a finite difference
equation for $\psi^{+}:\mathbb{Z}\to\mathbb{C}^{N}$. The edge index
can be characterized in terms of their behavior at $-\infty$. In fact
we have:
\begin{lem}
\label{lem:reformulation of edge index via non-normalizable bulk states}
\begin{align}
\mathcal{N}_{a} & = \dim V\,,\qquad V:=\{\psi^{+}:\mathbb{Z}\to\mathbb{C}^{N}|S\psi^{+}=0\text{ and }\,\psi_{n}^{+}\text{ is }\ell^{2}\text{ at }n\to-\infty\}\label{eq:reformulation of edge invariant}\,.
\end{align}
In particular $\mathcal{N}_{a}$ is seen to be independent of $a$,
independently of \cref{thm:BEC}. Moreover every $\psi\in V$ is uniquely determined by its restriction to $\mathbb{Z}_a$ for any $a$.
\begin{proof}
By \cref{eq:Fredholm index of S_a} we are led to determine the null spaces of $S_a$ and $S_a^\ast$ separately. The two operators act by \cref{eq:action of S in our model,eq:action of S* in our model} with $n\leq a$; \cref{eq:action of S in our model} comes without boundary conditions, whereas \cref{eq:action of S* in our model} is supplemented by $\psi_{a+1}^{-}=0$. By $B_{n}\in GL_N(\mathbb{C})$ the difference equation $S\psi^{+}=0$ can be solved recursively to the right ($\psi_{n-1}^{+}$ determines $\psi_{n}^{+}$) whereas $S^{\ast}\psi^{-}=0$ can be solved recursively to the left. Hence $\ker S_a=V$, whereas the Dirichlet boundary condition implies $\ker S_{a}^{\ast}=\{0\}$.
\end{proof}
\end{lem}

We next show that the edge index may be computed using a finite-box
truncation.

\begin{figure}[h]
\centerline{\begin{tikzpicture}[yscale=0.3]  
%	\draw [lightgray, fill=lightgray] (0, 0) rectangle (-4,4); 
%	\draw [dashed, lightgray, line width=4cm] (2, 4) -- (2, 5); 
%	\draw [dashed, lightgray, line width=4cm] (2, -4.05) -- (2, -5); 
%	\draw [lightgray, line width=0.2cm] (-4.2, 0) -- (-4.2, 4); 
%	\draw [lightgray, line width=0.2cm] (-4.5, 0) -- (-4.5, 4); 
%	\draw [lightgray, line width=0.2cm] (-4.8, 0) -- (-4.8, 4); 
	\draw [thin, -] (0,0) -- (-5.5,0); 
	\draw [thin, dashed] (0,0) -- (2,0);
%	\draw [thin, ->] (0,0) -- (0,5.5);  
%	\node [below right] at (5.5,0) {$\mathbb{Z}$}; 
%	\node [above] at (0,5.5) {$\hat{e}_2$}; 
%	\draw [gray, line width=0.2cm] (-0.01, 0) -- (-0.01, 4); 
%	\draw [dashed, gray, line width=0.2cm] (0.01, 4) -- (0.01, 5); 
%	\draw [dashed, gray, line width=0.2cm] (0.01, -5) -- (0.01, -4); 
%	\node [above] at (-2.5, 2) {Vacuum}; 
%	\node [above] at (2, 2) {Material}; 
%	\node [above, rotate=90] at (0, 2) {Edge}; 
	\draw [thin] (0,0) -- (0,-0.2);
	\node [below] at (0.7,-0.4) {$a\to+\infty$};
	\draw [thin] (-2,0) -- (-2,-0.2);
	\node [below] at (-2,-0.4) {$0$};
	\draw [{[-]}, ultra thick] (0,0) -- (-2,0);
	\node [above] at (-1,0.4) {$\operatorname{supp}\Lambda_a$};
\end{tikzpicture}}

\caption{\label{fig:Finite Window Used to compute the Edge's Index}The finite
box used in order to approximate the edge index.}
\end{figure}
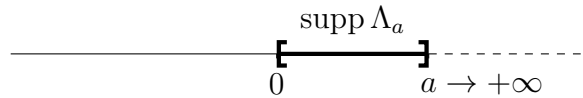
\begin{lem}
\label{lem:finite window regularization}The common value of $\mathcal{N}_{a}$, ($a\in\mathbb{Z}$)
is 
\begin{align}
\mathcal{N}^{\sharp} & =  \lim_{a\to+\infty}\tr(\Pi\Lambda_{a}P_{0,a})\label{eq:finite window regularization}\,.
\end{align}
\end{lem}

Here we denoted by $\Lambda_a$ the switch function $\Lambda$ when viewed as a multiplication operator
on $\ell^{2}(\mathbb{Z}_{a})$ and its descendant spaces. We have
\begin{align}
\iota_{a}\Lambda_{a}=\Lambda\iota_{a}\:,\qquad\Lambda_{a}\iota_{a}^{\ast}=\iota_{a}^{\ast}\Lambda\label{eq:commutation relations between J_a and Lambda}\,.
\end{align}
The switch $\Lambda_a$ roughly restricts states to $n\geq0$ within $n\leq a$, thereby singling out a finite box growing with $a$ (see \cref{fig:Finite Window Used to compute the Edge's Index}). The lemma asserts that edge states are unaffected by this restriction for $a\to+\infty$, because they are concentrated near the edge $n=a$. Consequently the task is to show $\|(\mathds{1}-\Lambda_a)P_{0,a}\|\to0$, $(a\to+\infty)$.

\begin{proof}[Proof of \cref{lem:finite window regularization}]
Let $V$ be the linear space of solutions $\psi=(\psi^+,\,0)$ seen in \cref{eq:reformulation of edge invariant}. By
\begin{align}
\dim V \leq N < \infty
\label{eq:finite dimensionality of V}
\end{align}
all norms on $V$ are equivalent, and we pick one, $\|\cdot\|_V$. For any $b\in\mathbb{Z}$, let $\mathcal{R}_b:V\to\mathcal{H}_{b}$ be defined by restriction. That map is injective by the conclusion of the previous lemma. Therefore and by \cref{eq:finite dimensionality of V} we have 
\begin{align}
\|\mathcal{R}_b\psi\| & \geq  c_{b}\|\psi\|\label{eq:image closed}
\end{align}
for some $c_b>0$. Elements $\psi\in V$ can not be $\ell^{2}$ at $n\to+\infty$ as
well, unless $\psi=0$, since that would imply a solution of $H\psi=0$,
which is ruled out by \cref{assu:zero_is_not_an_eigenvale_of_the_Hamiltonian}.
We thus have 
\begin{align}
\|\mathcal{R}_b\psi\|\to\infty\,,\qquad\,(b\to+\infty)\,.\label{eq:consequence of no zero eigenvalues}
\end{align}

For $b<a$ we denote by $\iota_{ab}:\ell^2(\mathbb{Z}_b)\hookrightarrow\ell^2(\mathbb{Z}_a)$ the injection (extension by zero); correspondingly $\iota_{ab}^\ast:\ell^2(\mathbb{Z}_a)\to\ell^2(\mathbb{Z}_b)$ is the restriction operator. These operators are analogous to those seen in \cref{eq:partial isometry condition}; in fact $\iota_b = \iota_a \iota_{ab}$. For $b$ large enough we have $(\mathds{1}-\Lambda)(1-\chi_b)=0$ by disjointness of support. Using \cref{eq:commutation relations between J_a and Lambda} we get $\mathds{1}-\Lambda=(\mathds{1}-\Lambda)\iota_b\iota_b^\ast = \iota_b(\mathds{1}-\Lambda_b)\iota_b^\ast$ and thus, by multiplication with $\iota_a^\ast$ and $\iota_a$ from left and right,
\begin{align}
\mathds{1}-\Lambda_a &= \iota_{ab} (\mathds{1}-\Lambda_b)\iota_{ab}^\ast
\label{eq:expression for Lambda_a^perp}
\end{align}
($b$ large, $a>b$).

Next we note that $P_{0,a}:\mathcal{H}_a\to\mathcal{H}_a$ induces a natural map $\mathcal{P}_{0,a}:\mathcal{H}_a\to V$, because for any $\psi_a\in\mathcal{H}_a$ the image $P_{0,a}\psi_a$ is the left tail of a solution in $V$ which it fully determines. It satisfies 
\begin{align}
\mathcal{R}_b\mathcal{P}_{0,a} & = \iota_{ab}^{\ast}P_{0,a}\,.\label{eq:exchange middle space between i_b and P_0a with V}
\end{align}

We next claim for any $b$
\begin{align}
\|\mathcal{R}_b\mathcal{P}_{0,a}\| & \to 0,\qquad(a\to+\infty)\,.\label{eq:edge zero modes live near the cut}
\end{align}
We have to show that $\|\mathcal{R}_b\mathcal{P}_{0,a}\psi_{a}\|\to0$
for every sequence $\psi_{a}\in\mathcal{H}_{a}$ with $\|\psi_{a}\|=1$.
Clearly the sequence at hand is at least bounded in $a$ (as well
as in $b>a$) and so is $\|\mathcal{P}_{0,a}\psi_{a}\|$ because of
\cref{eq:exchange middle space between i_b and P_0a with V} and
\cref{eq:image closed}. By compactness ($\dim V<\infty$) we have
$\mathcal{P}_{0,a}\psi_{a}\to\hat{\psi}$, ($a\to+\infty$) upon passing
to a subsequence. Hence $\mathcal{R}_b\mathcal{P}_{0,a}\psi_{a}$
has a limit $\mathcal{R}_b\hat{\psi}$ as $a\to+\infty$,
which inherits the boundedness in $b$. This contradicts \cref{eq:consequence of no zero eigenvalues}
unless $\hat{\psi}=0$, thus proving \cref{eq:edge zero modes live near the cut}. 

That in turn implies, by taking $b$ large and using \cref{eq:expression for Lambda_a^perp,eq:exchange middle space between i_b and P_0a with V},
\begin{align*}
\|(\mathds{1}-\Lambda_{a})P_{0,a}\|\to0\,,\qquad(a\to+\infty)\,.
\end{align*}
The same then holds in trace class norm $\|\cdot\|_{1}$ because
$\|A\|_{1}\leq\|A\|\operatorname{rank}A$ and $\operatorname{rank} P_{0,a}=\dim V$. Finally
\cref{eq:finite window regularization} follows by taking a (redundant)
limit of \cref{eq:edge invariant}.
\end{proof}

\begin{proof}[Proof of \cref{thm:reformulation of edge index in terms of the Lyapunov spectrum}] By \cref{eq:filtration of C^N by the Lyapunov spectrum} and the definition of $\chi_i$, as well as by \cref{assu:zero is not in the Lyapunov spectrum}, the RHS of \cref{eq:reformulation of edge index using Lyapunov spectrum} equals $\dim V_\chi = \dim V_0$ for some $\chi<0$. We also recall \cref{eq:finite difference equation formally extended to non normalizable vectors} and Definition  \cref{eq:reformulation of edge invariant}. We have the inclusions $V_\chi\subseteq V$ for any $\chi<0$ and $V\subseteq V_0$, since $\ell^2\subseteq\ell^\infty$ at $-\infty$. The conclusion now follows from \cref{lem:reformulation of edge index via non-normalizable bulk states}.
\end{proof}

\section{\label{sec:Proof-of-Duality}Proof of duality}

\begin{lem}
\label{lem:Edge projection strongly converges to bulk projection}
As $a\to+\infty$, \begin{align}
\iota_{a}P_{\pm,a}\iota_{a}^{\ast}-P_\pm & \stackrel{s}{\longrightarrow}  0\label{eq:Edge projection strongly converges to bulk projection}\,,\\
[\iota_{a}P_{\pm,a}\iota_{a}^{\ast}-P_\pm,\,\Lambda] & \stackrel{t}{\longrightarrow}  0\label{eq:commutator of projections with Lambda converges to zero}\,.
\end{align}
where $s,\,t$ denote strong and trace norm convergence respectively, and $P_{\pm,a}:=\chi_{(0,\infty)}(\pm H_a)$.
\end{lem}

\begin{proof}[Proof of \cref{thm:BEC}]
The operator $\Lambda_{a}$ introduced in \cref{eq:commutation relations between J_a and Lambda} is of finite
rank. The basic identity is 
\begin{align}
\tr(\Pi\Lambda_{a}) & =  0\,,\label{eq:Chirality is traceless}
\end{align}
which follows by evaluating the trace in the position basis and by
using $\tr_{\mathbb{C}^{2N}}\Pi=0$. We
insert $\mathds{1}=P_{0,a}+P_{+,a}+P_{-,a}$ with $P_{\pm,\,a}\equiv\chi_{(0,\,\infty)}(\pm H_{a})$
and obtain 
\begin{align}
\tr(\Pi\Lambda_{a}) & = \tr(\Pi\Lambda_{a}P_{0,a})+\tr(\Pi\Lambda_{a}P_{+,a})+\tr(\Pi\Lambda_{a}P_{-,a})\,.\label{eq:partition of unity}
\end{align}
The first term tends to $\mathcal{N}^{\sharp}$ as $a\to+\infty$
by \cref{eq:finite window regularization}. The second one is 
\begin{align*}
\tr(\Pi\Lambda_{a}P_{+,a}) & =  \tr(\Pi P_{-,a}\Lambda_{a}P_{+,a}) =  \tr(\Pi P_{-,a}[\Lambda_{a},\,P_{+,a}])\\
 & =  \tr(\iota_{a}\Pi P_{-,a}[\Lambda_{a},\,P_{+,a}]\iota_{a}^{\ast}) =  \tr(\Pi\iota_{a}P_{-,a}\iota_{a}^{\ast}[\Lambda,\,\iota_{a}P_{+,a}\iota_{a}^{\ast}])\,,
\end{align*}
where we used $\Pi P_{-,a}=P_{+,a}\Pi P_{-,a}$ (see \cref{eq:chiral symmetry constraint for edge Hamiltonian,eq:If H anti-commutes with chirality so does its odd functional calculus}), $\tr_{\mathcal{H}_{a}} A=\tr_{\mathcal{H}}(\iota_{a}A\iota_{a}^{\ast})$
for any operator $A$ on $\mathcal{H}_{a}$, as well as \cref{eq:partial isometry condition,eq:commutation relations between J_a and Lambda}. We next
use \cref{eq:Edge projection strongly converges to bulk projection,eq:commutator of projections with Lambda converges to zero}
together with the implication 
\begin{align}
X_{a}\stackrel{s}{\longrightarrow}X,\,Y_{a}\stackrel{t}{\longrightarrow}Y\quad\Longrightarrow\quad X_{a}Y_{a}\stackrel{t}{\longrightarrow}XY\,,\label{eq:trace norm convergence result from elbau graf 2002}
\end{align}
(see e.g. \cite{Elbau_Graf_2002}, Eq. (56)) to conclude 
\begin{align*}
\lim_{a\to+\infty}\tr(\Pi\Lambda_{a}P_{+,a}) & =  \tr(\Pi P_-[\Lambda,\,P_{+}])\,.
\end{align*}
We likewise have for the third term in \cref{eq:partition of unity}
\begin{align*}
\lim_{a\to+\infty}\tr(\Pi\Lambda_{a}P_{-,a}) & =  \tr(\Pi P_{+}[\Lambda,\,P_-])
\end{align*}
and thus find from \cref{eq:Chirality is traceless,eq:reformulation of bulk invariant}
that 
\begin{align*}
0 & =  \mathcal{N}^{\sharp}-\mathcal{N}\,.
\end{align*}
\end{proof}

In comparing the proofs of the cases of spectral and mobility gaps the following may be noted: While in the spectral gap case bulk and edge may be related at any finite $a$, in the mobility gap case the relation emerges at $a\to+\infty$, and this is made possible by \cref{lem:finite window regularization}.

\section{\label{sec:The Appendix}Appendix}

\subsection{Proofs of lemmas for the duality}

\begin{lem}
We have for $T$ operating on $\ell^{2}(\mathbb{Z})$
\begin{align}
\|T\|_{1} & \leq  \sum_{n,\,n'}\left|T(n,\,n')\right|\,.\label{eq:bound on trace class norm}
\end{align}
The bound is passed down to $\ell^2(\mathbb{Z})\otimes\mathbb{C}^{2N}$ provided $|\cdot|$ is interpreted as the trace norm of operators on the second factor.
\begin{proof}
Let $\mathcal{H}$ be a Hilbert space and $\{\varphi_{n}\}_{n}$
an orthonormal basis. Then 
\begin{align*}
\|T\|_{1}  \equiv  \sum_{n'}\left\langle \varphi_{n'},\,\left|T\right|\varphi_{n'}\right\rangle 
  \leq  \sum_{n'}\|\left|T\right|\varphi_{n'}\|
  =  \sum_{n'}\|T\varphi_{n'}\|\leq\sum_{n,n'}|\langle\varphi_{n},T\varphi_{n'}\rangle|\,,
\end{align*}
where we used $\|\psi\|\leq\sum_{n}|\langle\varphi_n,\psi\rangle|$  in the last step.

\end{proof}
\end{lem}

\begin{proof}[Proof of \cref{lem:=commutator of L and P is trace class}]
We first prove the trace class property. The operator $T=[\Lambda,\,P]$ has kernel 
\begin{align*}
T(n,\,n') & =  (\Lambda(n)-\Lambda(n'))P(n,\,n')\,.
\end{align*}
For large $\left|n\right|$ we have $\Lambda(n')=\Lambda(n)$ unless
$\left|n-n'\right|>\left|n\right|/2$. For such $n$ we have
\begin{align*}
\left|\Lambda(n)-\Lambda(n')\right|e^{-\mu\left|n-n'\right|}  \leq  2\|\Lambda\|_{\infty}e^{-\mu\left|n\right|/2}\leq  C(1+\left|n\right|)^{-\nu}
\end{align*}
for any $\mu,\,\nu>0$, suitable $C>0$ and all $n'$; thus 
\begin{align}
\left|\Lambda(n)-\Lambda(n')\right| & \leq  C(1+\left|n\right|)^{-\nu}e^{\mu\left|n-n'\right|}\,.\label{eq:equation which needs its constant to be adjusted}
\end{align}
For the finitely many remaining $n$ we have 
\[
\left|\Lambda(n)-\Lambda(n')\right|\leq2\|\Lambda\|_{\infty}\leq C(1+\left|n\right|)^{-\nu}
\]
by adjusting the constant $C$, and thus \cref{eq:equation which needs its constant to be adjusted}
as well. The bound $\|T\|_{1}<\infty$ now follows from \cref{assu:exp_decay_of_Fermi_projection} by \cref{eq:bound on trace class norm}. 

For the second statement of the lemma, we have by $\Sigma=P_+-P_-$ 
\begin{align}
2\mathcal{N} &= \tr\Pi P_{+}[\Lambda,\Sigma] - \tr \Pi P_- [\Lambda,\Sigma]\,.
\label{eq:beginning of reformulation of N}
\end{align}
The first term is the sum of two equal ones,
\begin{align*}
\tr\Pi P_{+}[\Lambda,\Sigma] = \tr\Pi P_{+}[\Lambda,P_{+}] - \tr\Pi P_{+}[\Lambda,P_-]  = -2\tr\Pi P_{+}[\Lambda,P_-]\,;
\end{align*}
indeed, by $P_{+}=(P_{+})^2$, and \cref{eq:rephrasing of the assumption of chiral symmetry} we have
\begin{align*}
\tr\Pi P_{+}[\Lambda,P_\pm] = \tr\Pi P_{+}[\Lambda,P_\pm] P_- 
= \mp\tr\Pi P_{+} \Lambda P_-\,.
\end{align*}
Likewise can be said about the last term in \cref{eq:beginning of reformulation of N}: \begin{align*}\tr\Pi P_-[\Lambda,\Sigma]=2\tr\Pi P_-[\Lambda,P_+]\,.\end{align*} We obtain \cref{eq:reformulation of bulk invariant}.
\end{proof}

\begin{proof}[Proof of \cref{lem:Edge projection strongly converges to bulk projection}]
We first prove \cref{eq:Edge projection strongly converges to bulk projection} and claim 
\begin{align*}
\iota_{a}f(H_{a})\iota_{a}^{\ast} & =  f(\iota_{a}H_{a}\iota_{a}^{\ast})+f\left(0\right)(\mathds{1}-\iota_{a}\iota_{a}^{\ast})
\end{align*}
for any (Borel) function $f$. In fact, let us decompose $\ell^{2}(\mathbb{Z})=\ell^{2}(\mathbb{Z}_{a})\oplus\ell^{2}(\tilde{\mathbb{Z}}_{a})$,
where $\tilde{\mathbb{Z}}_{a}=\mathbb{Z}\setminus\mathbb{Z}_a$,
as well as any descendant space such as $\mathcal{H}$. The isometries
$\iota_{a}$ and $\tilde{\iota}_{a}$ (similarly defined) provide
a partition of unity, $\mathds{1}=\iota_{a}\iota_{a}^{\ast}+\tilde{\iota}_{a}\tilde{\iota}_{a}^{\ast}$, and a block decomposition of 
\begin{align*}
\iota_{a}H_{a}\iota_{a}^{\ast}  =  \iota_{a}H_{a}\iota_{a}^{\ast}+\tilde{\iota}_{a}0\tilde{\iota}_{a}^{\ast}
  \equiv  H_{a}\oplus0\,.
\end{align*}
Thus, by the functional calculus, 
\begin{align}
f(\iota_{a}H_{a}\iota_{a}^{\ast}) & =  \iota_{a}f(H_{a})\iota_{a}^{\ast}+\tilde{\iota}_{a}f(0)\tilde{\iota}_{a}^{\ast}\,,
\label{eq:functional calculus on truncation of Hamiltonian}
\end{align}
as claimed.

For uniformly bounded operators, like $\iota_{a}H_{a}\iota_{a}^{\ast}$
and $H$, strong resolvent convergence is equivalent to strong convergence
(see \cite{Reed_Simon_FA_0125850506}, Problem VIII.28). The latter,
\[
\iota_{a}H_{a}\iota_{a}^{\ast}-H\stackrel{s}{\longrightarrow}0,\qquad(a\to+\infty)
\]
is evident, because the LHS vanishes for large but finite
$a$, when applied to any state $\psi\in\mathcal{H}$ from the dense
subspace $\{\supp\,\psi\subseteq\mathbb{Z}\text{ is bounded}\}$.
Finally we specialize to $f=\chi_{(-\infty,0)}$. By (\cite{Reed_Simon_FA_0125850506}, Theorem VIII.24 (b)) and \cref{assu:zero_is_not_an_eigenvale_of_the_Hamiltonian} the strong resolvent convergence implies $f(\iota_a H_a \iota_a^\ast)-f(H)\stackrel{s}{\to}0,\,(a\to+\infty)$. The limit \cref{eq:Edge projection strongly converges to bulk projection} now follows from \cref{eq:functional calculus on truncation of Hamiltonian} by $f(0)=0$.
\end{proof}

\begin{proof}[Proof of \cref{eq:commutator of projections with Lambda converges to zero}]
We write $D_{a}:=\iota_{a}P_{-,a}\iota_{a}^{\ast}-P_-$ for brevity. As
shown in (\cite{EGS_2005}, Eq. (3.20)), \cref{assu:exp_decay_of_Fermi_projection}
implies 
\begin{align*}
\|e^{-\mu n}e^{-\varepsilon\left|n\right|}P_-e^{\mu n}\| & \leq  C_{\varepsilon},\qquad(\varepsilon>0)\,,
\end{align*}
where $g(n)$ denotes the multiplication operator by the namesake function. The same holds true by the same assumption
for $\iota_{a}P_{-,a}\iota_{a}^{\ast}$ instead of $P_-$, and thus for
$D_{a}$ as well. The same estimate holds for $\mu$ replaced by $-\mu$. 

We pick a switch function $\Lambda$ with compactly supported variation and denote by $\Lambda^b(n)=\Lambda(n-b)$ its translate by $b\in\mathbb{N}$. We note that $\Lambda-\Lambda^b$ is of finite rank and that for fixed $\varepsilon\in\left(0,\,\mu\right)$ we have 
\begin{align*}
\|(\mathds{1}-\Lambda)e^{\mu n}e^{\varepsilon\left|n\right|}\| & \leq  C\,,\\
\|e^{-\mu n}\Lambda^b\|_{1} & \leq  Ce^{-\mu b}\,.
\end{align*}
The LHS of \cref{eq:commutator of projections with Lambda converges to zero} is 
\begin{align}
[D_{a},\,\Lambda] & =  \left(\mathds{1}-\Lambda\right)D_{a}\Lambda-\Lambda D_{a}\left(\mathds{1}-\Lambda\right)\label{eq:decomposition of commutator with switch to two summands-1}
\end{align}
and we claim that in the limit $a\to+\infty$ each term vanishes separately in trace norm. Indeed,
\begin{align*}
(\mathds{1}-\Lambda)D_{a}\Lambda & =  (\mathds{1}-\Lambda)D_{a}\Lambda^{b}+(\mathds{1}-\Lambda)D_{a}(\Lambda-\Lambda^{b})\,,\\
(\mathds{1}-\Lambda)D_{a}\Lambda^{b} & =  (\mathds{1}-\Lambda)e^{\mu n}e^{\varepsilon\left|n\right|} \cdot e^{-\mu n}e^{-\varepsilon\left|n\right|}D_{a}e^{\mu n}\cdot  e^{-\mu n}\Lambda^{b}
\end{align*}
Thus $\|(\mathds{1}-\Lambda)D_{a}\Lambda^{b}\|_{1}$ can be made
arbitrarily small, uniformly in $a$, by first picking $b$ large.
Then $\|(\mathds{1}-\Lambda)D_{a}(\Lambda-\Lambda^{b})\|_{1}$ will
be small for large $a$ by \cref{eq:Edge projection strongly converges to bulk projection}
(see \cref{eq:trace norm convergence result from elbau graf 2002}). The other term on the RHS of \cref{eq:decomposition of commutator with switch to two summands-1}
is dealt with similarly.
\end{proof}

\subsection{More general boundary conditions\label{subsec:General B.C.}}

In this section we generalize \cref{thm:BEC} to (largely) arbitrary
boundary conditions. In the case of a spectral gap (see \cref{sec:The-Spectral-Gap-Case}) we implement them by relaxing \cref{eq:relation between S_a and S} to \begin{align}S_a:=\iota_a^\ast S \iota_a+S_{BC}\,,\label{eq:custom BC for S_a}\end{align} where $S_{BC}$ is any compact operator. Since the Fredholm index is invariant under compact perturbations, the change does not affect the edge index \cref{eq:Edge index is a Fredholm index} and \cref{thm:BEC} remains true.

In the mobility gap regime and in the context of the model with nearest neighbor hopping \cref{eq:chiral Hamiltonian,eq:action of S in our model,eq:action of S* in our model} more general boundary conditions are obtained by allowing $S_{BC}$ to affect only sites $a$ and $a-1$; we thus allow the hopping matrices $A_n, B_n$ of the boundary $n=a$ to become singular, whereas they remain regular for $n\leq a-1$. The edge Hamiltonian \cref{eq:chiral edge Hamiltonian} remains defined with $S_a$ as in \cref{eq:relation between S_a and S}. Thus $S_a$ acts as in \cref{eq:action of S in our model} for $n\leq a$; likewise does $S_a^\ast$ as in \cref{eq:action of S* in our model}, except for $n=a$ where $(S_a^\ast\psi^-)_a=B_a^\ast\psi_a^-$.
\begin{prop} 
The edge index $\mathcal{N}_a$ is the same for all boundary matrices ($A_a$, $B_a$). In particular it is the same as in \cref{eq:reformulation of edge invariant}.
\begin{example}The case of regular $A_a, B_a$ corresponds to the edge Hamiltonian discussed so far. In relation to \cref{fig:Chiral zigzag Model} it amounts to breaking the thin bond between dimers $a$ and $a+1$. To set $A_a=0$ amounts to further remove one more dimer; to set instead $B_a=0$ to break the thick bond of the last dimer.
\end{example}
\begin{proof}
By \cref{eq:Fredholm index of S_a} we have $\mathcal{N}_a=\dim V^+ - \dim V^-$ with
\begin{align}
\mathclap{V^\pm = \{\psi^\pm:\mathbb{Z}\to\mathbb{C}^N|\psi^\pm\text{ is }\ell^2\text{ at }n\to-\infty\text{ and satisfies }\cref{eq:equation for ker S_a}\text{, resp. }\cref{eq:equation for ker S_a part 1,eq:equation for ker S_a part 2}\}\,,}\notag\\
A_{n+1}\psi_n^++B_{n+1}\psi_{n+1}^+&=0\,,\qquad(n\leq a-1)\label{eq:equation for ker S_a}\\
A_{n+1}^\ast\psi_{n+1}^-+B_n^\ast\psi_n^-&=0\,,\qquad(n\leq a-1)\label{eq:equation for ker S_a part 1}\\
B_a^\ast\psi_a^-&=0\,.\label{eq:equation for ker S_a part 2}
\end{align}
(The first equation is $(S\psi^+)_n=0$ for $n\leq a$ after shifting the index by one.) 

Introducing $\varphi_n^- = B_n^\ast \psi_n^-$, the eqs. \cref{eq:equation for ker S_a,eq:equation for ker S_a part 1} are solved iteratively to the left for $n\leq a-2$ by
\begin{align*}
\psi_n^+ = T_{n+1}\psi_{n+1}^+\,,\quad\varphi_n^-=T_{n+1}^\circ \varphi_{n+1}^-
\end{align*}
with $T_n=-A_n^{-1}B_n$ and $M^\circ=(M^\ast)^{-1}$. In particular,
\begin{align*}
\langle\psi_n^+,\,\varphi_n^-\rangle = \langle T_{n+1}\psi_{n+1}^+,\,T_{n+1}^\circ\varphi_{n+1}^-\rangle = \langle\psi_{n+1}^+,\,\varphi_{n+1}^-\rangle\,,
\end{align*}
which means that the LHS is constant in $n\leq a-1$. Actually we have 
\begin{align}
\langle\psi_n^+,\,\varphi_n^-\rangle = 0\,,\quad(n\leq a-1)\label{eq:orthogonality condition between psi_n_plus and b_n_ast psi_n_minus}
\end{align}
because of the $\ell^2$-condition, and not by resorting to \cref{eq:equation for ker S_a part 2}. To sum up: Since $A_n,B_n\in GL_N(\mathbb{C})\,,\,(n\leq a-1)$ the solutions of \cref{eq:equation for ker S_a,eq:equation for ker S_a part 1} with $n\leq a-2$ are bijectively determined by $\psi_{a-1}^+,\varphi_{a-1}^-\in\mathbb{C}^N$; among them, those that are $\ell^2$ correspond to subspaces (independent of $A_a, B_a$) of complementary dimensions. This follows from \cref{eq:Oseledet spaces of the circ sequence} with $\chi=0$.

The claim now follows by applying the following lemma to \cref{eq:equation for ker S_a,eq:equation for ker S_a part 1,eq:equation for ker S_a part 2} for $n=a-1$ by identifying $A_a=A$, $B_a=B$, $\psi_{a-1}^+=\psi^+$, $\psi_a^+=\tilde{\psi}^+$, $B_{a-1}^\ast\psi_{a-1}^-=\psi^-$, $\psi_a^-=\tilde{\psi}^-$, and by using \cref{eq:orthogonality condition between psi_n_plus and b_n_ast psi_n_minus} for $n=a-1$.
\end{proof}
\begin{lem}
Let an orthogonal decomposition $\mathbb{C}^N=V^+\oplus V^-$ and matrices $A,B\in \mathrm{Mat}_N(\mathbb{C})$ be given. We consider the set of equations
\begin{align}
A\psi^++B\tilde{\psi}^+&=0\label{eq:first equation for S_a lemma}\,,\\
A^\ast\tilde{\psi}^-+\psi^-&=0\,,\quad B^\ast\tilde{\psi}^-=0\label{eq:second equation for S_a_star lemma}
\end{align}
in the the unknowns $\psi^\pm\in V^\pm$, $\tilde{\psi}^\pm\in\mathbb{C}^N$. Then
\begin{align}
\dim\{(\psi^+,\tilde{\psi}^+)|\cref{eq:first equation for S_a lemma}\}-\dim\{(\psi^-,\tilde{\psi}^-)|\cref{eq:second equation for S_a_star lemma}\} = \dim V^+\,.\label{eq:final equation for conclusion of custom BC}
\end{align}
\begin{proof}
Let $P$ be the orthogonal projection onto $V^+$, whence $\dim \im P = \dim V^+$. Then the dimensions on the LHS of \cref{eq:final equation for conclusion of custom BC} are unaffected upon supplementing \cref{eq:first equation for S_a lemma,eq:second equation for S_a_star lemma} with $P \psi^+ = \psi^+$ and $P \psi^-=0$ respectively, while solving for $(\psi^\pm,\tilde{\psi}^\pm)\in\mathbb{C}^N\oplus\mathbb{C}^N=\mathbb{C}^{2N}$. We are then left computing $$I:=\dim\ker T_+ - \dim \ker T_-$$ with $T_+:\mathbb{C}^{2N}\to\mathbb{C}^{2N}$, $T_-:\mathbb{C}^{2N}\to\mathbb{C}^{3N}$ given by $$ T_+ = \begin{pmatrix} \mathds{1}-P & 0 \\ A & B\end{pmatrix}\,,\quad T_-=\begin{pmatrix}P & 0 \\ \mathds{1} & A^\ast \\ 0 & B^\ast\end{pmatrix}\,.$$

Using that $$ T_-^\ast = \begin{pmatrix}P & \mathds{1} & 0 \\ 0 & A & B\end{pmatrix} \,,\quad\begin{pmatrix}P & \mathds{1}-P & 0 \\ 0 & A & B\end{pmatrix}$$ have the same range, we have $$ \dim \im T_-^\ast = \dim V^+ + \dim \im T_+\,;$$ using also that $\dim \ker T_- = 2N - \dim \im T_-^\ast$  we find $$ I = \dim \ker T_+ + \dim \im T_+ - 2N + \dim V^+ = \dim V^+ \, .$$
\end{proof}
\end{lem}
\end{prop}

\begingroup
\let\itshape\upshape
 \printbibliography
 \endgroup
\end{document}